\documentclass[aps,prd,twocolumn,amsmath,amssymb,longbibliography,nofootinbib]{revtex4-2}

\usepackage[T1]{fontenc}
\usepackage{lmodern}
\usepackage{microtype}
\usepackage{bm}
\usepackage{graphicx}
\usepackage{hyperref}
\usepackage{amsthm}
\usepackage{xcolor}

\definecolor{linkblack}{RGB}{20,20,20}
\definecolor{citeblue}{RGB}{0,70,140}
\definecolor{urlblue}{RGB}{0,90,120}

\hypersetup{
  colorlinks=true,
  linkcolor=linkblack,
  citecolor=citeblue,
  urlcolor=urlblue
}

\newtheorem{theorem}{Theorem}
\newtheorem{proposition}[theorem]{Proposition}
\newtheorem{corollary}[theorem]{Corollary}
\newtheorem{remark}[theorem]{Remark}
\newtheorem{lemma}[theorem]{Lemma}

\allowdisplaybreaks
\begin{document}

\title{Birkhoff rigidity from a covariant optical seed}

\author{D. A.~Easson}
\email{easson@asu.edu}
\affiliation{
Department of Physics \& Beyond Center for Fundamental Concepts in Science,
Arizona State University, Tempe, Arizona 85287-1504, USA}

\begin{abstract}
We present a local seed-to--Kerr--Schild route to Birkhoff rigidity in
four-dimensional spherical vacuum gravity. On the two-dimensional orbit space, the
areal radius $r$ determines a scalar $F:=-(\nabla r)^2$, and the reduced vacuum
equations imply $F(r)=1-2M/r$. We show that the normalized one-forms
$dr/F$ and $(*dr)/F$ are closed, so that the null combinations
$F^{-1}(dr\pm *dr)$ are exact null seed forms. Integrating these yields local
Eddington--Finkelstein coordinates in which the metric takes Kerr--Schild form
over a flat background. We then prove the corresponding uniqueness statement in
the stationary optical sector: spherical symmetry forces the inverse optical seed
$\mathcal R$ to equal $\pm r$, equivalently the optical seed
$\rho$ to equal $\mp 1/r$, and the resulting seed data reconstruct the
Schwarzschild family. Thus, Birkhoff rigidity is paired with a spherical
converse theorem in the stationary optical framework: Schwarzschild is the unique
spherically symmetric stationary vacuum Kerr--Schild geometry generated by a
nowhere-vanishing optical seed.
\end{abstract}

\maketitle

\section{Introduction}
\label{sec:intro}

Birkhoff's theorem is the local rigidity statement that every spherically symmetric
vacuum solution of Einstein's equations is locally Schwarzschild and therefore
admits an additional Killing symmetry~\cite{Birkhoff1923,Jebsen1921,Jebsen2005,Alexandrow1923,Eiesland1925,Deser:2004gi,Schleich:2009uj}. Many proofs and reformulations of this result
are known, including coordinate-free arguments, local proofs in adapted
coordinates, dual-null formulations built from radial null congruences, and
approaches organized around the Kodama/Misner--Sharp structure
\cite{Misner:1964je,Kodama:1979vn,Schmidt:1997mq,Schleich:2009uj,Maciel:2018tnc,Hayward:1994bu,Abreu:2010ru}. Related rigidity results appear in non-Lorentzian settings; for example, a Birkhoff-type theorem for Kleinian split-signature black holes was established in~\cite{Easson:2023ytf}. The novelty of the present paper lies in two linked statements: a local exact null seed-form route from spherical vacuum gravity to a stationary Kerr--Schild chart, and, using the stationary optical framework of~\cite{Easson:2026xod}, a converse statement showing that the real monopole seed uniquely generates the Schwarzschild geometry. We use the term ``seed'' for covariantly defined scalar or one-form data from which the null coordinates, congruence, and Kerr--Schild geometry are reconstructed. By Kerr--Schild form we mean a decomposition
\(g_{\mu\nu}=\eta_{\mu\nu}+2V n_\mu n_\nu\), where
\(\eta_{\mu\nu}\) is flat and \(n_\mu\) is null with respect to both
\(\eta_{\mu\nu}\) and \(g_{\mu\nu}\).

In the optical framework, the geometry is organized by a complex
optical seed \(\rho=-\theta-i\omega\), where \(\theta\) and \(\omega\) are the
expansion and signed twist of the stationary Kerr--Schild congruence, and whose
inverse \(\mathcal R=-1/\rho\) obeys an eikonal equation and reconstructs the
Kerr--Schild congruence. Schwarzschild is
generated by a real seed, while Kerr is generated by a genuinely complex one.

The route developed here begins on the two-dimensional orbit space
\(Q\). The reduced vacuum equations imply that the scalar \(F:=-(\nabla r)^2\)
depends only on the areal radius and satisfies the seed ODE derived in
Sec.~\ref{sec:orbit_seed}. On regions where \(dr\neq 0\) and \(F\neq 0\), the
normalized one-forms introduced in Sec.~\ref{sec:normalized_forms} are closed, and
their null combinations are exact. Integrating those exact null seed forms yields
local Eddington--Finkelstein coordinates and makes the Kerr--Schild structure
manifest. 

The present result is the spherical
rigidity statement internal to the stationary optical-seed framework: spherical symmetry collapses
the inverse optical seed to the real monopole branches \(\mathcal R=\pm r\),
equivalently \(\rho=\mp 1/r\), and the associated Kerr--Schild reconstruction
yields Schwarzschild.

This formulation touches several familiar viewpoints while remaining distinct from
each of them. It shares ingredients with coordinate-free arguments, dual-null
formulations, and Kodama-time constructions, but its organizing feature is the
exact null seed-form packaging of the reduced vacuum equations.

Our goal is not merely to recover Eddington--Finkelstein coordinates in a new
notation, but to show that, in spherical vacuum gravity they arise as integrals of
covariantly defined exact null seed forms extracted directly from the reduced
field equations on the orbit space. In this way, the Kerr--Schild structure
emerges as an intrinsic consequence of the orbit-space vacuum system itself,
rather than being imposed only after Schwarzschild has been recognized in a
preferred chart. This provides the connection to the stationary optical framework: the forward
orbit-space construction and the optical converse are linked by a common
seed-based organization of the geometry.

The paper is organized as follows. Sections~\ref{sec:orbit_seed}--\ref{sec:EF_KS}
derive the orbit-space seed equation, construct the exact null seed forms, and
integrate them to local Eddington--Finkelstein and Kerr--Schild form. In
Sec.~\ref{sec:collapse_seed} we prove the corresponding uniqueness statement in the
stationary optical sector and, using the local converse of Ref.~\cite{Easson:2026xod},
derive the associated spherical converse theorem. We conclude in
Sec.~\ref{sec:discussion} with a brief discussion.

\section{Orbit-space seed equations}
\label{sec:orbit_seed}

We work with mostly-minus signature and consider a spherically symmetric spacetime
with two-dimensional Lorentzian orbit space \(Q\), equipped with metric \(h_{ab}\),
so that
\begin{equation}
ds^2=h_{ab}(x)\,dx^a dx^b-r(x)^2\,d\Omega_2^2,
\label{eq:warped_metric}
\end{equation}
where \(r(x)\) is the areal radius and \(d\Omega_2^2\) is the unit round metric on
\(S^2\). All contractions, covariant derivatives, and Hodge duals below refer to
\(h_{ab}\).

Let \(\epsilon_{ab}\) denote the volume form on \(Q\), normalized by
\begin{equation}
\nabla_a\epsilon_{bc}=0,
\qquad
\epsilon_{ab}\epsilon^{ab}=-2.
\end{equation}
For any one-form \(\alpha_a\), define its Hodge dual by
\begin{equation}
(*\alpha)_a:=\epsilon_a{}^b\alpha_b.
\label{eq:hodge_def}
\end{equation}
On one-forms in two-dimensional Lorentzian signature one has
\begin{equation}
\alpha\cdot(*\alpha)=0,
\qquad
(*\alpha)^2=-\alpha^2,
\qquad
*(*\alpha)=\alpha.
\label{eq:hodge_identities}
\end{equation}
We will also use
\begin{equation}
d(*df)=-(\Box f)\,\epsilon
\label{eq:dstar_identity}
\end{equation}
for any scalar \(f\), and
\begin{equation}
\alpha\wedge *\beta = -(\alpha\cdot\beta)\,\epsilon
\label{eq:wedge_identity}
\end{equation}
for any one-forms \(\alpha,\beta\).

\begin{proposition}[Seed ODE]
\label{prop:seed_ode}
Let \eqref{eq:warped_metric} be vacuum,
\begin{equation}
R_{\mu\nu}=0,
\end{equation}
with nonconstant areal radius \(r\). Then on \(Q\),
\begin{align}
R_{ab}[h]-2r^{-1}\nabla_a\nabla_b r &= 0,
\label{eq:red_tensor}\\
1+(\nabla r)^2+r\Box r &= 0.
\label{eq:red_scalar}
\end{align}
Since \(Q\) is two-dimensional,
\begin{equation}
\nabla_a\nabla_b r=\lambda\,h_{ab},
\qquad
\lambda:=\frac12\Box r.
\label{eq:hess_trace}
\end{equation}
Defining
\begin{equation}
F:=-\,(\nabla r)^2,
\label{eq:F_def}
\end{equation}
one finds, locally wherever \(dr\neq 0\),
\begin{equation}
1-F-rF'(r)=0,
\label{eq:F_ode}
\end{equation}
and hence
\begin{equation}
F(r)=1-\frac{2M}{r}
\label{eq:F_solution}
\end{equation}
for some constant \(M\).
\end{proposition}

\begin{proof}
For the warped product \eqref{eq:warped_metric}, the vacuum Einstein equations
reduce to \eqref{eq:red_tensor} and \eqref{eq:red_scalar}. Since \(Q\) is
two-dimensional,
\begin{equation}
R_{ab}[h]=\frac12 R[h]\,h_{ab},
\end{equation}
so \eqref{eq:red_tensor} implies \eqref{eq:hess_trace}.

Now differentiate \eqref{eq:F_def}:
\begin{equation}
\nabla_a F
=
-2\nabla^b r\,\nabla_a\nabla_b r
=
-2\lambda\,\nabla_a r.
\end{equation}
Thus \(F=F(r)\) locally wherever \(dr\neq 0\), and
\begin{equation}
F'(r)=-2\lambda=-\Box r.
\label{eq:Fprime_minus_box}
\end{equation}
Substituting \eqref{eq:Fprime_minus_box} into \eqref{eq:red_scalar} gives
\eqref{eq:F_ode}, which integrates to \eqref{eq:F_solution}.
\end{proof}

Proposition~\ref{prop:seed_ode} reduces the local vacuum content of spherical
symmetry to the single scalar function \(F(r)\), which will organize the seed-form
construction in the following sections.

\section{Closed normalized seed forms}
\label{sec:normalized_forms}

The seed equation for \(F(r)\) has an immediate differential consequence.

\begin{lemma}[Closed normalized seed forms]
\label{lem:normalized_forms}
On any open set where \(dr\neq 0\) and \(F\neq 0\), define
\begin{equation}
dR_*:=\frac{dr}{F},
\qquad
dT:=\frac{*dr}{F}.
\label{eq:dRstar_dT}
\end{equation}
Then both one-forms are closed:
\begin{equation}
d(dR_*)=0,
\qquad
d(dT)=0.
\label{eq:dRstar_dT_closed}
\end{equation}
Hence, on a simply connected local region, there exist functions \(R_*\) and \(T\)
such that \eqref{eq:dRstar_dT} holds.
\end{lemma}

\begin{proof}
Since \(F=F(r)\),
\begin{equation}
dF=F'(r)\,dr.
\end{equation}
Therefore
\begin{equation}
d(dR_*)
=
d(F^{-1}dr)
=
-\,F^{-2}dF\wedge dr
=0.
\end{equation}
For \(dT\),
\begin{align}
d(dT)
&=
d(F^{-1}*dr)
\nonumber\\
&=
-\,F^{-2}dF\wedge *dr+F^{-1}d(*dr)
\nonumber\\
&=
-\,\frac{F'}{F^2}\,dr\wedge *dr-\frac{\Box r}{F}\,\epsilon.
\end{align}
Using \eqref{eq:wedge_identity},
\begin{equation}
dr\wedge *dr = -(dr)^2\,\epsilon = F\,\epsilon,
\end{equation}
and then \eqref{eq:Fprime_minus_box}, we obtain
\begin{equation}
d(dT)
=
-\,\frac{F'}{F}\,\epsilon-\frac{\Box r}{F}\,\epsilon
=0.
\end{equation}
Thus both forms are closed.
\end{proof}

\begin{corollary}[Diagonal orbit-space form]
\label{cor:static_orbit}
On the same local region, the orbit-space metric may be written as
\begin{equation}
h
=
F(r)\,\bigl(dT^2-dR_*^2\bigr)
=
F(r)\,dT^2-\frac{dr^2}{F(r)}.
\label{eq:static_orbit}
\end{equation}
\end{corollary}

\begin{proof}
From \eqref{eq:dRstar_dT} and \eqref{eq:hodge_identities},
\begin{align}
(dT)^2 &= \frac{(*dr)^2}{F^2} = \frac{1}{F}, \\
(dR_*)^2 &= \frac{(dr)^2}{F^2} = -\frac{1}{F}, \\
dT \cdot dR_* &= 0.
\end{align}
Hence
\begin{equation}
h^{-1}=F^{-1}\bigl(\partial_T^2-\partial_{R_*}^2\bigr),
\end{equation}
or equivalently,
\begin{equation}
h=F(r)\,\bigl(dT^2-dR_*^2\bigr).
\end{equation}
Since \(dr=F\,dR_*\), the second equality in \eqref{eq:static_orbit} follows immediately.
\end{proof}

\begin{corollary}[Kodama one-form and Killing symmetry]
\label{cor:Kodama}
Define the Kodama one-form on \(Q\) by
\begin{equation}
K^\flat:=*dr.
\label{eq:Kodama_oneform}
\end{equation}
Then
\begin{equation}
K^\flat = F\,dT,
\qquad
K^2=F,
\label{eq:Kodama_FdT}
\end{equation}
and the corresponding vector field \(K^a:=h^{ab}K_b\) is Killing:
\begin{equation}
\nabla_{(a}K_{b)}=0.
\label{eq:Kodama_Killing}
\end{equation}
\end{corollary}

\begin{proof}
The relation \(K^\flat=F\,dT\) follows directly from \eqref{eq:dRstar_dT}. For the
Killing property,
\begin{align}
\nabla_aK_b
&=
\nabla_a(\epsilon_b{}^c\nabla_c r)
=
\epsilon_b{}^c\nabla_a\nabla_c r
\nonumber\\
&=
\lambda\,\epsilon_b{}^c h_{ac}
=
\lambda\,\epsilon_{ba},
\end{align}
using \eqref{eq:hess_trace}. Hence
\begin{equation}
\nabla_{(a}K_{b)}=0.
\end{equation}
Finally,
\begin{equation}
K^2=(*dr)^2=-\,(dr)^2=F.
\end{equation}
\end{proof}

\section{Exact null seed forms}
\label{sec:null_seed_forms}

The null combinations of the normalized seed forms are the primary objects in our proof. They are the point at which the present formulation most clearly differs from
more standard static-coordinate, Kodama, and dual-null treatments of local
Birkhoff rigidity.

\begin{theorem}[Exact null seed forms in spherical vacuum]
\label{thm:exact_null_seed_forms}
On any simply connected local region of a nontrivial spherically symmetric vacuum
spacetime for which \(dr\neq 0\) and \(F\neq 0\), define
\begin{equation}
k_\pm := \frac{1}{F}\bigl(dr\pm *dr\bigr).
\label{eq:kpm_def}
\end{equation}
Then:
\begin{align}
k_\pm &= dR_* \pm dT,
\label{eq:kpm_TR}\\
k_\pm^2 &= 0,
\label{eq:kpm_null}\\
dk_\pm &= 0,
\label{eq:kpm_closed}\\
k_\pm\cdot dr &= -1.
\label{eq:kpm_dot_dr}
\end{align}
The associated null vectors
\begin{equation}
\ell_\pm^a:=h^{ab}(k_\pm)_b
\end{equation}
are geodesic and affinely parametrized,
\begin{equation}
\ell_\pm^b\nabla_b(\ell_\pm)_a=0,
\label{eq:geodesic}
\end{equation}
with
\begin{equation}
\ell_\pm^a\nabla_a r=-1.
\label{eq:r_affine}
\end{equation}
Thus \(r\) serves as an affine parameter along the radial null geodesics, up to
sign.
\end{theorem}

\begin{proof}
Equation \eqref{eq:kpm_TR} is immediate from \eqref{eq:dRstar_dT}. Closure
\eqref{eq:kpm_closed} then follows from Lemma~\ref{lem:normalized_forms}.

To prove nullness, use \eqref{eq:hodge_identities} and \((dr)^2=-F\):
\begin{equation}
(dr\pm *dr)^2
=
(dr)^2+(*dr)^2\pm 2\,dr\cdot *dr
=
-F+F=0,
\end{equation}
hence \(k_\pm^2=0\). Likewise,
\begin{equation}
k_\pm\cdot dr
=
\frac{1}{F}\bigl((dr)^2\pm *dr\cdot dr\bigr)
=
-\frac{F}{F}
=-1.
\end{equation}

Since \(dk_\pm=0\), \(\nabla_a(k_\pm)_b=\nabla_b(k_\pm)_a\). Therefore
\begin{align}
\ell_\pm^b\nabla_b (k_\pm)_a
&=
\ell_\pm^b\nabla_a (k_\pm)_b
=
\frac12\nabla_a(k_\pm^2)
=
0.
\end{align}
Since \((\ell_\pm)_a=(k_\pm)_a\), this is \eqref{eq:geodesic}. Finally,
\begin{equation}
\ell_\pm^a\nabla_a r=k_\pm\cdot dr=-1,
\end{equation}
which is \eqref{eq:r_affine}.
\end{proof}

\begin{corollary}[Double-null orbit-space form]
\label{cor:double_null_orbit}
Let
\begin{equation}
v:=T+R_*,
\qquad
u:=T-R_*,
\label{eq:uv_defs}
\end{equation}
so that
\begin{equation}
dv=k_+,
\qquad
du=-\,k_-.
\end{equation}
With this convention, \(v\) is the advanced null coordinate and \(u\) is the retarded null coordinate.

Then the orbit-space metric may be written equivalently as
\begin{equation}
h=F(r)\,\bigl(dT^2-dR_*^2\bigr)=F(r)\,du\,dv.
\label{eq:double_null_orbit}
\end{equation}
\end{corollary}

\begin{proof}
From \eqref{eq:uv_defs},
\begin{equation}
dT=\frac12(dv+du),
\qquad
dR_*=\frac12(dv-du).
\end{equation}
Then
\begin{equation}
dT^2-dR_*^2=du\,dv,
\end{equation}
which gives \eqref{eq:double_null_orbit}.
\end{proof}

\section{Emergent Eddington--Finkelstein and Kerr--Schild structure}
\label{sec:EF_KS}

Integrating the exact null seed forms yields local Eddington--Finkelstein
coordinates and a local Kerr--Schild decomposition over a flat background.

\begin{theorem}[From exact null seed forms to local Eddington--Finkelstein and Kerr--Schild form]
\label{thm:EF_KS}
On any simply connected local region where \(dr\neq 0\) and \(F\neq 0\), the exact
null seed forms of Theorem~\ref{thm:exact_null_seed_forms} integrate to local null
coordinates for which the metric takes the ingoing Eddington--Finkelstein form
\begin{equation}
ds^2
=
F(r)\,dv^2-2\,dv\,dr-r^2d\Omega_2^2,
\qquad
F(r)=1-\frac{2M}{r},
\label{eq:ingoing_EF}
\end{equation}
or the outgoing form
\begin{equation}
ds^2
=
F(r)\,du^2+2\,du\,dr-r^2d\Omega_2^2.
\label{eq:outgoing_EF}
\end{equation}
Equivalently, in the ingoing chart,
\begin{equation}
g_{\mu\nu}
=
\eta_{\mu\nu}+2V\,n_\mu n_\nu,
\qquad
V=-\frac{M}{r},
\qquad
n_\mu dx^\mu=dv,
\label{eq:KS_main}
\end{equation}
with our mostly-minus signature convention, the Schwarzschild Kerr--Schild profile
appears as \(V=-M/r\).
Here, the flat background is
\begin{equation}
\eta
=
dv^2-2\,dv\,dr-r^2d\Omega_2^2.
\label{eq:flat_bg}
\end{equation}
Hence the spacetime is locally Schwarzschild.
\end{theorem}

\begin{proof}
Starting from \eqref{eq:static_orbit}, use \(v=T+R_*\) and \(dr=F\,dR_*\):
\begin{align}
h
&=
F(dT^2-dR_*^2)
\nonumber\\
&=
F(dv-dR_*)^2-F\,dR_*^2
\nonumber\\
&=
F\,dv^2-2F\,dv\,dR_*
\nonumber\\
&=
F\,dv^2-2\,dv\,dr.
\end{align}
This gives \eqref{eq:ingoing_EF}. The outgoing form follows analogously from
\(u=T-R_*\).

Now define
\begin{equation}
\eta:=dv^2-2\,dv\,dr-r^2d\Omega_2^2.
\end{equation}
Setting
\begin{equation}
t:=v-r
\end{equation}
gives
\begin{equation}
\eta=dt^2-dr^2-r^2d\Omega_2^2,
\end{equation}
so \(\eta\) is Minkowski space in spherical coordinates. Therefore
\begin{equation}
g-\eta
=
\bigl(F(r)-1\bigr)\,dv^2
=
-\frac{2M}{r}\,dv^2.
\end{equation}
Since \(dv\) is null with respect to \(\eta\), this is a Kerr--Schild decomposition,
\begin{equation}
g_{\mu\nu}
=
\eta_{\mu\nu}+2V\,n_\mu n_\nu,
\qquad
V=-\frac{M}{r},
\qquad
n_\mu dx^\mu=dv.
\end{equation}
Thus the local Kerr--Schild form follows directly from the exact seed forms,
without first identifying the local Schwarzschild geometry in a preferred
chart.
\end{proof}

\section{Collapse to the real monopole seed}
\label{sec:collapse_seed}

The orbit-space construction of Secs.~\ref{sec:orbit_seed}--\ref{sec:EF_KS}
reaches Schwarzschild directly. The purpose of the present section is complementary:
to record the corresponding uniqueness and converse statements internal to the
stationary optical sector itself. The first statement concerns the stationary
optical scalar system on a flat background; the second combines it with the
converse machinery of Ref.~\cite{Easson:2026xod}.

The same symbol \(r\) is used below for the orbit-space areal radius and for
the Euclidean radial coordinate \(r=\sqrt{x^2+y^2+z^2}\) on the flat
Kerr--Schild background. In the spherically symmetric reconstruction these
quantities coincide, since the symmetry spheres have area \(4\pi r^2\).

\begin{proposition}[Spherical uniqueness of the stationary optical seed]
\label{prop:seed_collapse_monopole}
Let \(U\subset \mathbb R^3\setminus\{0\}\) be a connected \(SO(3)\)-invariant
region, with origin chosen at the center of symmetry. Let
\(\mathcal R:U\to\mathbb C\) be a nowhere-vanishing function and define
\begin{equation}
\rho:=-\frac{1}{\mathcal R}.
\end{equation}
Assume that \(\mathcal R\) is spherically symmetric,
\begin{equation}
\mathcal R=\mathcal R(r),
\qquad
r:=\sqrt{x^2+y^2+z^2},
\label{eq:R_radial}
\end{equation}
and that it satisfies the stationary optical scalar system
\begin{equation}
(\nabla\mathcal R)^2=1,
\qquad
\mathcal R\,\nabla^2\mathcal R=2,
\label{eq:spherical_optical_scalar_system}
\end{equation}
where \(\nabla\) and \(\nabla^2\) are the flat Euclidean gradient and Laplacian on
\(\mathbb R^3\). Then there exists \(\sigma\in\{+1,-1\}\) such that
\begin{equation}
\mathcal R=\sigma r,
\qquad
\rho=-\frac{\sigma}{r}.
\label{eq:seed_collapse_result}
\end{equation}
Thus, up to the choice of null branch, the stationary optical seed collapses to the
unique real monopole.

Moreover, the reconstructed spatial congruence is
\begin{equation}
\mathbf{k}=\sigma\,\nabla r,
\end{equation}
and the corresponding optical scalars are
\begin{equation}
\theta=\frac{\sigma}{r},
\qquad
\omega=0.
\end{equation}
\end{proposition}

\begin{proof}
Because \(\mathcal R=\mathcal R(r)\), the eikonal equation in
\eqref{eq:spherical_optical_scalar_system} gives
\begin{equation}
\left(\frac{d\mathcal R}{dr}\right)^2=1.
\label{eq:dRdr_squared}
\end{equation}
Since \(d\mathcal R/dr\) is continuous and takes values in the discrete set
\(\{+1,-1\}\), it is constant on each connected radial interval. Hence there exist
\(\sigma\in\{+1,-1\}\) and \(c\in\mathbb C\) such that
\begin{equation}
\mathcal R(r)=\sigma r+c.
\label{eq:R_linear_radial}
\end{equation}

For any radial function \(f(r)\) on \(\mathbb R^3\setminus\{0\}\),
\begin{equation}
\nabla^2 f=f''+\frac{2}{r}f'.
\label{eq:radial_laplacian}
\end{equation}
Applying this to \eqref{eq:R_linear_radial} gives
\begin{equation}
\nabla^2\mathcal R=\frac{2\sigma}{r}.
\end{equation}
Hence
\begin{equation}
\mathcal R\,\nabla^2\mathcal R
=
(\sigma r+c)\frac{2\sigma}{r}
=
2+\frac{2\sigma c}{r}.
\end{equation}
Imposing the second equation in \eqref{eq:spherical_optical_scalar_system} forces
\begin{equation}
c=0.
\end{equation}
Therefore
\begin{equation}
\mathcal R=\sigma r,
\qquad
\rho=-\frac{1}{\mathcal R}=-\frac{\sigma}{r},
\end{equation}
which proves \eqref{eq:seed_collapse_result}.

Because \(\mathcal R\) is real and \((\nabla\mathcal R)^2=1\), the stationary
reconstruction reduces to
\begin{equation}
\mathbf{k}=\nabla\mathcal R=\sigma\,\nabla r.
\end{equation}
Finally, since \(\rho=-\theta-i\omega\), one finds
\begin{equation}
\theta=\frac{\sigma}{r},
\qquad
\omega=0.
\end{equation}
\end{proof}

The orbit-space derivation up to this point is self-contained. The converse
statement below is conditional only on the local reconstruction formalism of
Ref.~\cite{Easson:2026xod}, which we use in the minimal form stated explicitly in
Theorem~\ref{thm:spherical_optical_converse}.

\begin{theorem}[Spherical converse in the stationary optical sector]
\label{thm:spherical_optical_converse}
Let \(\mathcal R\) be a nowhere-vanishing spherically symmetric solution of the
stationary optical scalar system on a connected \(SO(3)\)-invariant region
\(U\subset\mathbb R^3\setminus\{0\}\), and reconstruct the stationary optical data
and vacuum Kerr--Schild metric by the local converse of
Ref.~\cite{Easson:2026xod}. Then
\begin{equation}
\mathcal R=\sigma r,
\qquad
\rho=-\frac{\sigma}{r},
\qquad
\mathbf{k}=\sigma\,\nabla r,
\qquad
\sigma\in\{+1,-1\},
\end{equation}
and the reconstructed stationary vacuum Kerr--Schild metric is locally
Schwarzschild. Equivalently, up to the choice of null branch, the only
spherically symmetric solution of the stationary optical system reconstructs the
Schwarzschild family.
\end{theorem}

\begin{proof}
By Proposition~\ref{prop:seed_collapse_monopole},
\begin{equation}
\mathcal R=\sigma r,
\qquad
\rho=-\frac{\sigma}{r},
\qquad
\mathbf{k}=\sigma\,\nabla r,
\qquad
\omega=0.
\end{equation}
The converse theorem of Ref.~\cite{Easson:2026xod} then reconstructs the
stationary vacuum Kerr--Schild data with
\begin{equation}
V=m\,\Re \rho=-\frac{m\sigma}{r}
\end{equation}
for some real constant \(m\). Defining
\begin{equation}
M:=m\sigma,
\end{equation}
one obtains
\begin{equation}
V=-\frac{M}{r}.
\end{equation}
The associated null Kerr--Schild one-form is
\begin{equation}
n_\mu dx^\mu = dt+\sigma\,dr,
\end{equation}
so the reconstructed metric is
\begin{equation}
g_{\mu\nu}
=
\eta_{\mu\nu}+2V\,n_\mu n_\nu
=
\eta_{\mu\nu}-\frac{2M}{r}\,n_\mu n_\nu.
\end{equation}
Since the flat background is
\begin{equation}
\eta = dt^2-dr^2-r^2d\Omega_2^2,
\end{equation}
this gives
\begin{equation}
ds^2
=
dt^2-dr^2-r^2d\Omega_2^2
-\frac{2M}{r}(dt+\sigma\,dr)^2.
\label{eq:spherical_converse_metric}
\end{equation}
For \(\sigma=+1\), Eq.~\eqref{eq:spherical_converse_metric} is the ingoing
Eddington--Finkelstein Kerr--Schild form of Schwarzschild; for \(\sigma=-1\), it is
the outgoing form. Hence the reconstructed stationary vacuum Kerr--Schild metric is
locally Schwarzschild.
\end{proof}

For the ingoing branch singled out by Theorem~\ref{thm:EF_KS}, one has
\(\sigma=+1\), and therefore \(\mathcal R=r\) and \(\rho=-1/r\). In this sense, the spherical vacuum sector isolates a common monopole seed whose
distinct classical dressings yield the Schwarzschild/Newton potential on the
gravitational side and the Coulomb potential on the gauge side.

\begin{remark}[Locality and horizons]
\label{rem:locality_horizons}
The construction is local and applies on any region where \(dr\neq 0\) and
\(F=1-2M/r\neq 0\). At a nondegenerate horizon \(F=0\), the normalized seed forms
\(dT\) and \(dR_*\) become singular, but the corresponding ingoing and outgoing
Eddington--Finkelstein charts extend regularly across the horizon. Global extensions
are obtained by patching neighboring local regions.
\end{remark}

\section{Discussion}
\label{sec:discussion}

We have presented a local seed-to--Kerr--Schild route to Birkhoff rigidity in
four-dimensional spherical vacuum gravity. The reduced orbit-space vacuum equations
give rise to exact null seed forms whose integration yields the local
Eddington--Finkelstein and Kerr--Schild descriptions of Schwarzschild. The stationary optical sector provides the complementary converse statement: spherical symmetry collapses the inverse optical seed to
the real monopole branches, and the associated Kerr--Schild reconstruction is
Schwarzschild.  The same monopole profile admits a natural double-copy interpretation:
gravitational Kerr--Schild dressing yields Schwarzschild, while electrostatic
dressing yields Coulomb.

Several questions remain open. Beyond the spherical case, it would be
interesting to see if the present covariant null seed-form route, arising
directly from the reduced orbit-space vacuum equations, admits a comparably
direct analogue in less symmetric settings, and in particular beyond the
stationary sector in which genuinely complex seeds already occur, as in the
Kerr solution. It is also natural to ask whether the rigidity interpretation
isolated in the spherical case, which singles out the real monopole sector,
extends in a useful form to genuinely complex stationary seeds. A particularly
immediate extension concerns the \(\Lambda\neq 0\) case: one may ask whether the
present covariant seed-form route admits a Schwarzschild--(A)dS analogue, and
whether its unique spherical branch matches the curved-background deformation of
the optical-seed picture.

\acknowledgments
It is a pleasure to thank M. Falato and M. Pezzelle for valuable discussions. DAE is supported in part by the U.S. Department of Energy,
Office of High Energy Physics, under Award Number DE-SC0019470.

\bibliographystyle{apsrev4-2}

\begin{thebibliography}{99}

\bibitem{Birkhoff1923}
G.~D.~Birkhoff,
\textit{Relativity and Modern Physics}
(Harvard University Press, Cambridge, MA, 1923), p.~253.

\bibitem{Jebsen1921}
J.~T.~Jebsen,
``Über die allgemeinen kugelsymmetrischen Lösungen der Einsteinschen Gravitationsgleichungen im Vakuum,''
Ark.\ Mat.\ Astron.\ Fys.\ (Stockholm) \textbf{15}, no.~18, 1--9 (1921).

\bibitem{Jebsen2005}
J.~T.~Jebsen,
``On the general spherically symmetric solutions of Einstein's gravitational equations in vacuo,''
Gen.\ Relativ.\ Gravit.\ \textbf{37}, 2253--2259 (2005),
doi:10.1007/s10714-005-0168-y.

\bibitem{Alexandrow1923}
W.~Alexandrow,
``Über den kugelsymmetrischen Vakuumvorgang in der Einsteinschen Gravitationstheorie,''
Ann.\ Phys.\ (Leipzig) \textbf{377}, no.~18, 141--152 (1923),
doi:10.1002/andp.19233771804.

\bibitem{Eiesland1925}
J.~Eiesland,
``The group of motions of an Einstein space,''
Trans.\ Amer.\ Math.\ Soc.\ \textbf{27}, no.~2, 213--245 (1925),
doi:10.1090/S0002-9947-1925-1501308-7.



\bibitem{Schleich:2009uj}
K.~Schleich and D.~M.~Witt,
``A simple proof of Birkhoff's theorem for cosmological constant,''
J. Math. Phys. \textbf{51}, 112502 (2010)
doi:10.1063/1.3503447
[arXiv:0908.4110 [gr-qc]].

\bibitem{Deser:2004gi}
S.~Deser and J.~Franklin,
``Schwarzschild and Birkhoff a la Weyl,''
Am. J. Phys. \textbf{73}, 261-264 (2005)
doi:10.1119/1.1830505
[arXiv:gr-qc/0408067 [gr-qc]].

\bibitem{Misner:1964je}
C.~W.~Misner and D.~H.~Sharp,
``Relativistic equations for adiabatic, spherically symmetric gravitational collapse,''
Phys. Rev. \textbf{136}, B571-B576 (1964)
doi:10.1103/PhysRev.136.B571

\bibitem{Kodama:1979vn}
H.~Kodama,
``Conserved Energy Flux for the Spherically Symmetric System and the Back Reaction Problem in the Black Hole Evaporation,''
Prog. Theor. Phys. \textbf{63}, 1217 (1980)
doi:10.1143/PTP.63.1217

\bibitem{Schmidt:1997mq}
H.~J.~Schmidt,
``A New proof of Birkhoff's theorem,''
Grav. Cosmol. \textbf{3}, 185-190 (1997)
[arXiv:gr-qc/9709071 [gr-qc]].



\bibitem{Maciel:2018tnc}
A.~Maciel, M.~Le Delliou and J.~P.~Mimoso,
``Revisiting the Birkhoff theorem from a dual null point of view,''
Phys. Rev. D \textbf{98}, no.2, 024016 (2018)
doi:10.1103/PhysRevD.98.024016
[arXiv:1803.11547 [gr-qc]].


\bibitem{Hayward:1994bu}
S.~A.~Hayward,
``Gravitational energy in spherical symmetry,''
Phys. Rev. D \textbf{53}, 1938-1949 (1996)
doi:10.1103/PhysRevD.53.1938
[arXiv:gr-qc/9408002 [gr-qc]].

\bibitem{Abreu:2010ru}
G.~Abreu and M.~Visser,
``Kodama time: Geometrically preferred foliations of spherically symmetric spacetimes,''
Phys. Rev. D \textbf{82}, 044027 (2010)
doi:10.1103/PhysRevD.82.044027
[arXiv:1004.1456 [gr-qc]].

\bibitem{Easson:2023ytf}
D.~A.~Easson and M.~W.~Pezzelle,
``Kleinian black holes,''
Phys. Rev. D \textbf{109}, no.4, 044007 (2024)
doi:10.1103/PhysRevD.109.044007
[arXiv:2312.00879 [hep-th]].

\bibitem{Easson:2026xod}
D.~A.~Easson and M.~J.~Falato,
``Untwisting the double copy: the zeroth copy as an optical seed,''
JHEP \textbf{05}, 249 (2026)
doi:10.1007/JHEP05(2026)249
[arXiv:2604.05103 [hep-th]].

\end{thebibliography}

\end{document}